\documentclass[sigconf,nonacm]{acmart}

\usepackage[capitalise]{cleveref}
\usepackage{setspace}
\usepackage[ruled,linesnumbered,noend]{algorithm2e}
\usepackage{setspace}
\usepackage[noend]{algpseudocode}

\newtheorem{theorem}{Theorem}
\newtheorem{lemma}[theorem]{Lemma}

\usepackage{thm-restate}

\newtheorem{definition}[theorem]{Definition}

\newtheorem{remark}[theorem]{Remark}

\newtheorem{claim}[theorem]{Claim}

\newcommand{\Eq}{\mathsf{Equality}}
\newcommand{\Com}{\mathsf{CommitteeElect}}
\newcommand{\localcom}{\mathsf{LocalCommitteeElect}}
\newcommand{\sparse}{\mathsf{SparseNetwork}}
\newcommand{\nulll}{\mathsf{Null}}
\newcommand{\gossip}{\mathsf{Gossip}}

\newcommand{\negl}{\mathsf{negl}}
\newcommand{\polylog}{\mathsf{polylog}}
\newcommand{\poly}{\mathsf{poly}}
\newcommand{\Otilde}{\tilde{O}}

\newcommand{\Real}{\mathsf{Real}}
\newcommand{\Ideal}{\mathsf{Ideal}}
\newcommand{\Sim}{\mathsf{Sim}}
\newcommand{\Adv}{\mathsf{Adv}}
\newcommand{\cZ}{\mathcal{Z}}
\newcommand{\cF}{\mathcal{F}}
\newcommand{\cH}{\mathcal{H}}
\newcommand{\SB}{\mathsf{SB}}

\newcommand{\out}{\mathsf{Out}}
\newcommand{\PKE}{\mathsf{PKE}}
\newcommand{\SKE}{\mathsf{SKE}}
\newcommand{\DS}{\mathsf{DS}}
\newcommand{\Gen}{\mathsf{Gen}}
\newcommand{\Genenc}{\mathsf{Gen_{enc}}}
\newcommand{\Gensig}{\mathsf{Gen_{sig}}}
\newcommand{\Sign}{\mathsf{Sign}}
\newcommand{\Vrfy}{\mathsf{Vrfy}}
\newcommand{\ct}{\mathsf{ct}}
\newcommand{\pk}{\mathsf{pk}}
\newcommand{\sk}{\mathsf{sk}}

\newcommand{\Comp}{\mathsf{Comp}}

\newcommand{\Enc}{\mathsf{Enc}}

\newcommand{\Dec}{\mathsf{Dec}}

\newcommand{\suppress}[1]{}

\newcommand{\under}[2]{\underbrace{#1}_{\substack{#2}}}

\renewcommand{\epsilon}{\varepsilon}

\SetKwInput{KwInput}{Input}
\SetKwInput{KwOutput}{Output}

\algdef{SE}[SUBALG]{Indent}{EndIndent}{}{\algorithmicend\ }%
\algtext*{Indent}
\algtext*{EndIndent}

\usepackage{xspace}

\makeatletter
\DeclareRobustCommand\onedot{\futurelet\@let@token\@onedot}
\def\@onedot{\ifx\@let@token.\else.\null\fi\xspace}

\def\eg{\emph{e.g}\onedot} 
\def\ie{\emph{i.e}\onedot}

\makeatother

\crefname{algocf}{algorithm}{algorithms}
\Crefname{algocf}{Algorithm}{Algorithms}
\crefname{claim}{Claim}{Claims}

\AtBeginDocument{%
  }

\begin{document}

\title{On the Communication Complexity of Secure Multi-Party Computation With Aborts}

\author{James Bartusek}
\orcid{0009-0000-9020-3589}
\affiliation{
  \institution{UC Berkeley}
  \city{Berkeley}
  \state{California}
  \country{USA}
}
\email{jamesbartusek@berkeley.edu}

\author{Thiago Bergamaschi}
\orcid{0009-0006-5579-8090}
\affiliation{
  \institution{UC Berkeley}
  \city{Berkeley}
  \state{California}
  \country{USA}
}
\email{thiagob@berkeley.edu}

\author{Seri Khoury}
\orcid{0000-0002-7491-5866}
\affiliation{
  \institution{UC Berkeley}
  \city{Berkeley}
  \state{California}
  \country{USA}
}
\email{seri_khoury@berkeley.edu}

\author{Saachi Mutreja}
\orcid{0009-0002-6825-7112}
\affiliation{
  \institution{Columbia University}
  \city{New York}
  \state{New York}
  \country{USA}
}
\email{saachi@berkeley.edu}

\author{Orr Paradise}
\orcid{0000-0001-8212-8558}
\affiliation{
  \institution{UC Berkeley}
  \city{Berkeley}
  \state{California}
  \country{USA}
}
\email{orrp@eecs.berkeley.edu}




\begin{abstract}
  A central goal of cryptography is Secure Multi-party Computation (MPC), where $n$ parties desire to compute a function of their joint inputs without letting any party learn about the inputs of its peers. Unfortunately, it is well-known that MPC guaranteeing output delivery to every party is infeasible when a majority of the parties are malicious. In fact, parties operating over a point-to-point network (\ie, without access to a broadcast channel) cannot even reach an \emph{agreement} on the output when more than one third of the parties are malicious (Lamport, Shostak, and Pease, JACM 1980).\\


  Motivated by this infeasibility in the point-to-point model, Goldwasser and Lindell (J. Cryptol 2005) introduced a definition of MPC that \emph{does not require agreement}, which today is generally referred to as MPC \emph{with selective abort}. Under this definition, any party may abort the protocol if they detect malicious behavior. They showed that MPC with selective abort is feasible for any number of malicious parties by implementing a broadcast functionality with abort. \\

  While the model of MPC with abort has attracted much attention over the years, little is known about its communication complexity over point-to-point networks. In this work, we study the communication complexity of MPC with abort and devise nearly-optimal communication efficient protocols in this model. Namely, we prove trade-offs between the number of honest parties $h$, the communication complexity, and the \emph{locality} of the protocols. Here, locality is a bound on the number of peers with which each party must communicate. Our results are as follows:

    \vspace{-3pt}
  
  \begin{enumerate}
      \item \textbf{Near-optimal communication:} A protocol with $\Otilde(n^2/h)$ communication complexity. 
      \item \textbf{Near-optimal locality:} A protocol with $\Otilde(n^3/h)$ communication complexity where each party communicates with only $\Otilde(n/h)$ other parties. 
      \item \textbf{Intermediate communication vs locality:} A protocol with communication complexity $\Otilde(n^3/h^{3/2})$ where each party communicates with only $\Otilde(n/\sqrt{h})$ parties. 
      \item \textbf{Lower bound:} An $\Omega(n^2/h)$ lower bound on the communication complexity. In particular, we show that any party must communicate with $\Omega(n/h)$ other parties. Our lower bound is inspired by a recent lower bound for Broadcast by Blum et. al (DISC 2023).
  \end{enumerate}
\end{abstract}




\maketitle

\section{Introduction}
Secure multi-party computation (MPC) \cite{Yao1986HowTG, Goldreich1987HowTP, BenOr1988CompletenessTF, Chaum1988MultipartyUS} is a fundamental primitive in cryptography. In an MPC protocol, $n$ parties desire to compute a joint function of their inputs, without letting any party learn about the input of another---other than what can already be learned from the output. Significant effort has been devoted to understanding how many bits must be communicated by parties during the execution of the protocol (hereafter \emph{communication complexity}) \cite{Franklin1992CommunicationCO, DamgardN07, Damgrd2010PerfectlySM, Asharov2012MultipartyCW, Damgrd2012MultipartyCF}.




MPC protocols are often constructed in a model where \emph{broadcast} is available, which enables each party to send a message to all other parties, and guarantees that the same message will be heard by all receiving parties. The assumption that broadcast is available has played a crucial role in several classic MPC results. However, in practice, networks commonly consist only of \emph{point-to-point} (also known as peer-to-peer) connections. In this point-to-point model, a malicious party may claim to have sent a message to all parties while only sending it to some. Achieving broadcast over point-to-point networks (without assuming a trusted preprocessing phase) is possible only when at most one third of the parties are malicious \cite{10.1145/357172.357176, 10.1145/323596.323602}. 

Motivated by this infeasibility in the point-to-point model, Goldwasser and Lindell \cite{Goldwasser2005SecureMC} introduced the notion of \emph{MPC with selective abort} (hereafter \emph{MPC with abort}). MPC with abort is a relaxation that allows each party to either correctly compute the function, or abort if malicious behavior is detected. As discussed in \cite{Goldwasser2005SecureMC}, it becomes possible to guarantee that all honest parties either correctly compute the function or abort, even when a majority of parties are malicious, and even when no broadcast channel is available. We are interested in the communication complexity of MPC in this setting (\ie, no trusted preprocessing, no broadcast channel, and no bound on the number of adversaries). We ask:
\begin{center}
    \textit{How much must the parties communicate to correctly compute $f$, or detect malicious behavior?}
\end{center} 

Note that Goldwasser and Lindell \cite{Goldwasser2005SecureMC} study the \emph{round complexity} of MPC with abort---that is, the number of rounds until $f$ is computed or the parties abort. However, not much is understood about MPC with abort in terms of \textit{communication complexity}. Indeed, even a single round of communication in their protocol may require $\Omega(n^2)$ communicated bits.

Our focus is on the behavior of the MPC protocol as the number of parties $n$ scales. In light of this, we try to optimize not just the communication complexity, but also the \emph{locality} of the protocol \cite{BoyleGT13}, \ie, the number of parties with which a single party communicates. Any protocol is trivially $(n-1)$-local, but better locality of, say, $\polylog(n)$ is desirable as $n$ grows, since each party then needs to secure and maintain far fewer channels.

\subsection{Contributions}

Following \cite{Goldwasser2005SecureMC}, our goal is to study how the communication complexity of MPC with abort over point-to-point networks scales as a function of both the total number of parties $n$, and (a lower bound on) the number of honest parties $h$ in the network. While we allow for the very basic setup of a shared common random string (CRS), we do not assume strong trusted preprocessing such as public key infrastructure (PKI). For simplicity, in the ensuing discussion we focus on computing constant-depth functions $f$ on $n$ constant-sized inputs with a single constant-sized output and omit the dependence on the security parameter. The security of all of our protocols relies on computational assumptions, namely, the Learning with Errors assumption \cite{Regev09}\footnote{Refer to \cref{subsection:prelim-mpc} for a description of parameters.}. Finally, we stress again that, for us, ``MPC with abort'' refers to MPC with selective abort over point-to-point channels. We defer formal statements and generalizations to the body.\footnote{We denote $g(n) = \Tilde{O}(f(n))$ if there exists a constant $c$ such that $g(n) \leq c f(n)\polylog f(n)$.}

Our first two contributions are two novel protocols achieving state-of-the-art communication complexity and locality (respectively):

\begin{theorem}\label{theorem:results_new_protocol}
    There exists a protocol for MPC with abort against static malicious adversaries using $\Tilde{O}(n^2/h)$ bits of communication, where $n$ is the number of parties of which at least $h$ are honest.
\end{theorem}

\begin{theorem}\label{theorem:results-optimal-locality}
There exists a protocol for MPC with abort against static malicious adversaries using $\Tilde{O}(n^3/h)$ bits of communication, with locality $\Tilde{O}(n/h)$.
\end{theorem}

We complement these protocols with nearly-matching lower bounds, which imply that \Cref{theorem:results_new_protocol,theorem:results-optimal-locality} are optimal for communication complexity and locality (respectively). Our proof is inspired by an ``indistinguishability argument", similar to a recent lower bound for Broadcast by Blum et al.~\cite{BlumBCL23}, which holds even in broader settings. 

\begin{restatable}{theorem}{thmLB}
\label{theorem:lowe-bound}
    Any protocol for MPC with abort against malicious adversaries requires $\Omega(n^2/h)$ bits of communication. Furthermore, such protocols must have locality $\Omega(n/h)$.
\end{restatable}

Finally, although \Cref{theorem:results-optimal-locality} achieves near-optimal locality, it comes at a cost in terms of communication complexity over \Cref{theorem:results_new_protocol}. A natural question is whether this trade-off is inherent, or whether it is possible to sacrifice some locality to obtain better communication complexity.
We answer this question in the positive, by combining ideas from \Cref{theorem:results-optimal-locality} with our new communication-efficient protocol of \Cref{theorem:results_new_protocol} to obtain our final main result:

\begin{theorem}
    \label{theorem:results-local-tradeoff}
     There exists a protocol for MPC with abort against static malicious adversaries using $\Tilde{O}(n^3/h^{3/2})$ bits of communication, with locality $\Tilde{O}(n/\sqrt{h})$.
\end{theorem}

\subsection{Prior work}\label{sec:PriorWork}


The focus of this work is on how the communication complexity of MPC scales with the number of parties and tolerated adversaries, which has a long and rich history of study in various models.

There has been significant effort towards optimizing the communication complexity of full-fledged Byzantine agreement (\ie, without abort) \cite{Dolev82,DolevR85,KingS09,KingS11,Braud-SantoniGH13}. These works study the problem with either no setup, or assuming a public-key infrastructure (PKI) where each party begins the protocol with knowledge of a public key associated with each other party.

There have also been many works that study the more general question of communication complexity of \emph{MPC with guaranteed output delivery} (\ie, without abort), either in the information-theoretic setting (\eg \cite{HirtMP00,DamgardN07}), or in the computational setting (\eg \cite{BoyleGT13,BoyleCG24}).

In our work, we study the relaxed model of Goldwasser and Lindell \cite{Goldwasser2005SecureMC} where the computation is considered secure even if some parties abort in the presence of adversaries. This notion of \emph{security with abort} is now ubiquitous in the MPC literature, and is the de facto standard model for studying MPC in the dishonest majority setting \cite{KatzOS03,Asharov2012MultipartyCW,MukherjeeW16,10.1145/3566048,BenhamoudaL18,ChoudhuriCG0O20}. Some of these works (\eg \cite{Asharov2012MultipartyCW}) study the communication complexity of MPC with abort, however, to the best of our knowledge, no prior work has been dedicated to optimizing the communication complexity of MPC with abort \emph{over point-to-point networks}.
We also remark that, starting with \cite{doi:10.1137/0212045}, there has been a line of work devoted to studying the communication complexity of \emph{broadcast} protocols over point-to-point networks in the dishonest majority setting, including recent progress on upper and lower bounds \cite{CohenDKS23}. However, their setting is different from ours in that their end goal is broadcast (rather than general MPC), and they assume the existence of public-key infrastructure (PKI) between parties.


Lastly, we also study communication complexity in the setting of \emph{locality} \cite{BoyleGT13}, where each party is may communicate with only a limited number of other parties. Prior work \cite{BoyleGT13} has also considered this question, but \emph{only} in the honest majority setting, and with strong setup assumptions (\eg PKI).


\subsection{Organization}

We organize the remainder of this work as follows. In \Cref{sec:Tech}, we provide a technical overview of our new protocols. In \Cref{sec:preli}, we provide preliminaries and basic definitions. In \Cref{sec:upper}, we present our communication-optimal protocol for MPC with abort of \Cref{theorem:results_new_protocol}. Subsequently, in \Cref{sec:lower}, we present our lower bounds. Finally, in \Cref{section:local}, we present our protocols with locality.

\section{Technical overview}\label{sec:Tech}
\label{sec:techGL}

The starting point for our work, and a key building block in our later constructions, is a simple protocol for MPC with abort with $\tilde{O}(n^2)$ communication complexity which works for all $h<n$. In \Cref{subsection:overview-goldwasser-lindell} we describe how small modifications to the protocol of Goldwasser and Lindell \cite{Goldwasser2005SecureMC} achieve this bound. In the ensuing  \Cref{subsection:overview-optimal_communication}, we describe how to ``bootstrap'' this result using subsampling techniques and stronger cryptographic primitives to achieve our result in \Cref{theorem:results_new_protocol}. Finally, in \Cref{subsection:overview-local}, we describe our protocols with locality.


\subsection{The case of \texorpdfstring{$h<n$}{h<n}}
\label{subsection:overview-goldwasser-lindell}

We begin by outlining the protocol for MPC with abort introduced by Goldwasser and Lindell~\cite{Goldwasser2005SecureMC}. Their original protocol requires $O(n^3)$ bits of communication, and leverages the observation that MPC can be achieved by a constant number of invocations of All-to-All Broadcast (see also~\cite{CrepeauGT95,MukherjeeW16} and references therein). In All-to-All Broadcast, every party wants to broadcast a message to all the other parties. Although Broadcast is impossible if the fraction of honest parties is smaller than $2/3$ \cite{Lamport1982TheBG}, Broadcast \textit{with abort} is possible for any number of honest parties.  \\

\noindent \textbf{Single-source broadcast with abort \cite{Goldwasser2005SecureMC}}
\begin{enumerate}
    \item \textbf{Broadcast step:} A party $P$ that wants to broadcast a message $m$, sends $m$ to all other parties.
    \item \textbf{Verification step:} Each party sends to all other parties the message that it received from $P$ in the first step.
    \item \textbf{Output step:} If a party received two different messages in the second step, it aborts. Otherwise, if a party got the same message from all other parties in the second step, it outputs the value of that message.
\end{enumerate}

Observe that the protocol uses all the links in the network, incurring $O(n^2)$ communication. \\

\noindent \textbf{All-to-all broadcast with abort} The All-to-All Broadcast with abort protocol of~\cite{Goldwasser2005SecureMC} simply runs $n$ invocations of the above single-source Broadcast with abort protocol in parallel. Thus, it requires $O(n^3)$ communication, and implies the same communication complexity for MPC.\\

\noindent \textbf{Succinct equality testing and efficient verification} We begin with a simple observation which allows us to ``shave-off'' a factor of $n$ on the communication complexity of All-to-All Broadcast (with aborts). Which, in turn, gives us a simple $\Tilde{O}(n^2)$ protocol for MPC (with aborts) for any $h<n$. Informally, we optimize the verification step above by concatenating the messages from the $n$ parallel runs of Broadcasts with abort, and use hash functions to pairwise test the equality between these strings using only $O(\log(n))$ bits of communication (see \Cref{sec:equality}). While this protocol still communicates on all edges/links, it communicates only $O(\log n)$ bits on each edge, resulting in $O(n^2\log n)$ bits total.


\subsection{Protocols with optimal communication complexity}
\label{subsection:overview-optimal_communication}

\textbf{The $\Theta(n^2)$ barrier} Goldwasser and Lindell \cite{Goldwasser2005SecureMC} rely on All-to-All Broadcast to achieve MPC. Clearly, since there are $\Theta(n^2)$ edges in the network, any protocol for All-to-All Broadcast requires $\Omega(n^2)$ bits of communication. One may wonder whether it is actually possible to overcome this communication barrier. At first glance, to achieve MPC with abort on point-to-point networks, it may seem that each pair of parties needs to \emph{verify} that they have consistent views of the rest of the network. Without such verification, the adversary could potentially mislead the parties causing them to arrive to different outputs \textit{without aborting}.\\  


\noindent \textbf{Going below $\Theta(n^2)$} We overcome this issue by allowing honest parties to verify their views through other honest parties. At a high level, our protocol picks a random subset of the nodes which we refer to as \emph{the committee}, and delegates the MPC task to this committee.\footnote{This approach is inspired by committee-based consensus protocols used in large-scale blockchains, albeit applied to a highly non-standard (strong) dishonest majority setting. The idea of leveraging a committee can be traced back to \cite{BrachaT84} (PODC 84) in an application to distributed deadlock detection; see \cite{Li2020ABD, Xu2023ASO} for a modern review. }  Electing this subset in a communication-efficient manner, as well as delegating the computation itself, require overcoming several challenges. While the former is addressed via a combination of techniques from \Cref{subsection:overview-goldwasser-lindell}, overcoming the latter requires additional insight described next.

To delegate the computation we leverage an encrypted functionality akin to Threshold Fully Homomorphic Encryption (TFHE) \cite{Asharov2012MultipartyCW, Boneh2018ThresholdCF}. Informally, the committee members generate public key and secret key pairs, where the secret key is $k$-out-of-$k$ secret shared among the parties. So long as there is at least one honest party in the committee, no information about the secret key---and consequently, about the honest parties inputs---is revealed to the adversary. 


 The relevant details on the encrypted functionality are provided in \Cref{subsection:prelim-mpc}, and a formal description of our protocol is presented in \Cref{sec:upper}. In what remains of this subsection, we provide an informal description of the main steps in our protocol. If at any step of our protocol a verification step fails, the parties immediately abort. We omit this detail from the high-level description below for conciseness, but include it in the formal description in \Cref{sec:upper}. \\
 
\noindent \textbf{Main steps of our protocol}
 \begin{enumerate}
     \item A committee of size $k = O(n h^{-1}\log n)$ is selected uniformly at random via a ``self-election'' scheme. By a ``Hitting-Set'' argument, at least $1$ honest party is elected. 
     \item The committee members notify the rest of the network of their election, and subsequently verify that they each have consistent views of the other committee members. 
     \item The committee creates a public/secret key pair $(pk,sk)$, where the secret key is secret-shared into $(sk_j)_{j\in [k]}$.
     \item Each party in the committee sends the public key $pk$ to all other parties.
     \item All $n$ parties encrypt their inputs using $pk$, and send their ciphertexts to the committee members.
     \item Using MPC, the committee members compute the output of the function from the encrypted inputs and the secret key shares.
     \item Finally, the committee members forward the output to the other parties.
 \end{enumerate}

 The full description of the protocol, as well as its formal details, are provided in \Cref{sec:upper}.\\

\noindent \textbf{Communication complexity} In the protocol above, each party only communicates with the committee. As there are only $O(n\log n/h)$ parties in the committee, and since each other party only sends and receives $O(\log n)$ bits to and from the committee (as we show in \Cref{sec:upper}), it follows that the total communication complexity is $O(n^2\polylog n/h)$.

\subsection{Protocols with locality}
\label{subsection:overview-local}

We build on our communication-efficient protocols for MPC with abort by designing protocols that achieve locality. \\

\noindent \textbf{Establishing a sparse routing network}
The starting point for our local protocols is the observation that our protocol for All-to-All Broadcast with abort can be sparsified. Indeed, if the parties were given a sparse communication network on which to communicate, with the guarantee that all the honest parties are connected, then at least in principle they could ensure consistent views of each other's inputs. 

To establish this sparse routing network, each party locally samples $d = \Theta(n h^{-1}\log n)$ nodes at random, and attempts to establish them as their ``next hop'' in the network. To ensure locality, parties must be wary not to accept too many incoming connections, which is why we set up a \emph{bidirectional} network (\ie, the ``next hop'' relation is symmetric). Indeed, if any party detects too many incoming connections, it notifies its next hops and aborts. The threshold for ``too many'' is chosen such that, with all but negligible probability, it is surpassed only if the party was maliciously targeted by the adversary (a distributed denial of service attack, if you will).\\
    
    
    
\noindent \textbf{Responsible routing on the network}
It is tempting to declare victory (\ie, a local protocol for simultaneous broadcast with abort): after all, all honest parties are connected, so they may simply broadcast their message over the routing network. Not quite, since the honest parties do not know each other's identities. Furthermore, without any public-key infrastructure, the adversary may forge messages as they please; that is, if party $P$'s next hop $X$ is claiming to be forwarding a message originating in some source party $S$ ($S\to X \to P$), party $P$ has no way of ascertaining that the message indeed originated with $S$.
    
This is overcome using a similar strategy to that employed by Goldwasser and Lindell \cite{Goldwasser2005SecureMC}: when party $P$ hears (via one of its hops) that party $S$'s input is $x$, it forwards this information to all of its other hops---unless it has \emph{already heard} that $S$'s input was some other $x'\neq x$ (an equivocation), in which case it warns all its hops, and aborts. Naturally, any party that receives such warning will then forward it and abort. We refer to this technique as \emph{responsible gossip}, since parties spread rumors about other parties across the network, but only if the rumors are not contradictory.
    
Total communication is kept low by requiring each honest party $P$ to forward a rumor (``$S$ has input $x$'') at most once. Thus, honest parties each send at most $n + 1$ different messages (accounting for the warning message) to each of their $O(d)$ hops, for a total of $O(n^2\cdot d) = O(n^3h^{-1}\log n)$ communication per simultaneous broadcast. Since All-to-All Broadcast (with aborts) is sufficient to achieve MPC (with aborts), this all but concludes our description of the proof of \Cref{theorem:results-optimal-locality}.\\

\noindent \textbf{Local committee election} While conceptually simple, \Cref{theorem:results-optimal-locality} incurs a serious blowup in communication complexity. It is then natural to ask whether one can combine our network sparsification techniques with our committee-based, communication-optimal protocols, in order to establish tradeoffs between locality and communication complexity. 

To do so, the starting point is to leverage our ``responsible gossip'' protocol to perform the committee election \emph{locally}. That is, the committee members $C\subset [n]$ announce their self-election by sending their ID via the routing network, using $O(|C|\cdot d\cdot n)$ bits total. While again it may seem tempting to claim victory, unfortunately simply electing the committee is not enough to decrease communication: for the parties' (encrypted) inputs to even reach the committee using the routing network, each of $n$ different inputs would have to be sent over the $O(d\cdot n) = \Tilde{O}(n^2h^{-1})$ edges of the graph, which regresses to our previous bound of $\Tilde{O}(n^3h^{-1})$.\\

\noindent \textbf{Sparsifying the committee--network communication, and a covering claim}
Instead, after the committee is elected, we entirely dispense with the routing network, and attempt to sparsify the committee--network interaction. To do so, the committee members partition the network by independently sampling (overlapping) subsets of the parties $S_u \subset [n]$ of size $s$, for each $u\in C$, which they are ``responsible'' for. So long as each party in the network is connected to---or ``covered'' by---at least one honest committee member, then its input will correctly reach the committee. To ensure every committee member contains the encrypted input of every party, the committee members forward to each other the inputs they received from their subsets $S_u$. Once the honest committee members have a consistent view of all the parties' (encrypted) inputs, the actual computation can be correctly delegated. 

Unfortunately, to ensure the honest committee members $C\cap H$ ``cover'' the entire network, we need to increase the committee size. In particular, so long as $s\cdot |C\cap H| \approx n\cdot \log n$, a hitting set argument ensures the covering claim. The resulting locality of the protocol is then bounded by the sizes $s, |C|, $ and the degree of the routing graph. In turn, the total communication complexity dominated by 
\begin{gather}
    \under{O(|C|\cdot d\cdot n)}{\text{Local Committee} \\ \text{Election}}
    \ \ +
    \under{\Tilde{O}(|C|^2\cdot s)}{\text{Committee--Network} \\ \text{Interaction}}
    +
    \under{\Tilde{O}(|C|^2)}{\text{Delegating} \\ \text{Computation}}  \\
    \leq \tilde{O}\big(|C|n^2 h^{-1}+ |C|^2\cdot s\big).
\end{gather}

\noindent Balancing the choice of $|C|$ and $s$ subject to the covering claim reveals the optimal choice of $|C| = s = \Tilde{O}(nh^{-1/2})$, which gives us the advertised bounds of \Cref{theorem:results-local-tradeoff}.


\section{Preliminaries}\label{sec:preli}

\subsection{Model, basic definitions, and assumptions}


\paragraph{The network} $n$ parties want to compute a joint functionality $\cF$ on their private inputs in a synchronized point-to-point network. That is, any pair of parties may communicate directly via a communication channel. We assume parties have a lower bound on the number of honest parties $h$ in the network; the tighter this bound, the more efficient the resulting protocol.

\paragraph{The adversary} We consider a \textit{static} and \emph{malicious} adversary. Of the $n$ parties, the adversary may choose before the protocol begins to corrupt up to $n-h$ parties, and we refer to the set of at least $h$ honest parties as $\cH$. The adversary may deviate from the protocol in an attempt to trick the honest parties into computing a faulty result, or to learn something about their inputs. 

\paragraph{Communication complexity} Our main measure of protocol efficiency is the number of bits communicated between parties during computation. However, as discussed in \cite{BoyleGT13}, an adversary may flood the network with messages which, intuitively, should not be counted towards the communication complexity of a protocol. To avoid this issue, we define the communication complexity of a protocol to be the total number of bits sent by all parties \emph{if they were to honestly follow the protocol}. For a randomized protocol, we take the worst-case over all possible executions (in the all-honest case). Indeed, throughout this work, we will assume that each honest party aborts if it were to receive more bits than prescribed by the protocol.


\paragraph{The security parameter $\lambda$} Throughout the paper, we use $\lambda$ to denote the security parameter, which is a variable that quantifies the level of security provided by a cryptographic algorithm, typically represented as the size (in bits) of the keys used.

\subsection{Succinct equality testing}\label{sec:equality}

In this section we state the following folklore lemma, which says that two parties with strings of length $n$ bits each, can detect whether their strings are equal with high probability while exchanging only $O(\log n)$ bits. 

\begin{lemma}\label{lemma:equality}
  Fix $\lambda\in \mathbb{N}$. Two parties with inputs $m_1, m_2\in \{0, 1\}^n$ can detect whether $m_1$ and $m_2$ are equal, using a randomized protocol which exchanges only $O(\lambda\log n)$ bits and such that:
  \begin{enumerate}
      \item If $m_1 = m_2$, then the protocol outputs 1.
      \item If $m_1\neq m_2$, then the protocol outputs 0 with probability at least $1-\frac{1}{n^\lambda}$.
  \end{enumerate}

\end{lemma}

\begin{algorithm}[ht]
    \setstretch{1.35}
    \caption{$\Eq_\lambda$, Succinct Equality Test.}
    \label{alg:equality}
    \KwInput{Parties $P_1, P_2$ with inputs $m_1, m_2\in \{0, 1\}^n$.}
    \KwOutput{A flag $f\in \{0, 1\}$.}

    \begin{algorithmic}[1]
    \State $P_1$ samples a prime $p\in[n^\lambda]$ uniformly at random, and sends to $P_2$ the value of $p$ and
    \begin{equation*}
        v_1 \coloneqq m_1 \mod p.
    \end{equation*}
    
    \State $P_2$ outputs and sends $1$ to $P_1$ if and only if $v_1 = m_2 \mod p$. Else, it outputs and sends 0. 

    \end{algorithmic}\label{protocol:Equality}

\end{algorithm}

For a proof of \Cref{lemma:equality}, see \eg \cite{Lovett,Pitassi}.

\subsection{Secure multiparty computation}
\label{subsection:prelim-mpc}

To formally define security, we follow the standard universal composability framework \cite{Canetti01}. We consider a probabilistic polynomial-time (PPT) \emph{environment} $\cZ$ that is invoked on the security parameter $1^\lambda$ and auxiliary input $z$, and conducts the protocol execution in one of two worlds. In the \emph{real world}, $\cZ$ initializes parties $P_1,\dots,P_n$ with inputs $x_1,\dots,x_n$ and initializes an adversary $\Adv$ that corrupts up to $n-h$ of the parties. The remaining honest parties execute a protocol $\Pi$ with the corrupted parties controlled by $\Adv$, who may interact arbitrarily with $\cZ$. At the end of the protocol, all honest parties and the adversary send output to $\cZ$. In the \emph{ideal world}, $\cZ$ initializes ``dummy'' parties $\widetilde{P}_1,\dots,\widetilde{P}_n$ that are supposed to simply forward their inputs to a trusted third party and forward their outputs back to $\cZ$, and initializes an adversary $\Sim$ that corrupts up to $n-h$ of the dummy parties. The remaining honest parties and $\Sim$ execute a protocol with the trusted third party who is implementing an ``ideal functionality'' for $\cF$.

Clearly, in the ideal execution, the simulator does not learn anything besides their output of the function. Thus, if the environment cannot distinguish between the real world execution and the ideal world execution, then the real world adversary must have also learned nothing besides their output of the function.


\paragraph{Real execution} The environment $\cZ(1^\lambda,z)$ initializes parties \linebreak $P_1,\dots,P_n$ with $x_1,\dots,x_n$, and initializes a PPT $\Adv$ corrupting parties $[n] \setminus \cH$. The protocol $\Pi$ is executed in the presence of $\Adv$. The honest parties follow the instructions of the protocol $\Pi$, and the adversary is allowed to deviate from it while communicating with $\cZ$. At the end of the protocol, $\cZ$ receives outputs from the honest parties and output from $\Adv$. Finally, it computes its own output. This output is denoted by the random variable $\Real_{\Pi,\Adv,\cZ}(1^\lambda,z)$.



\paragraph{Ideal execution} We describe the ideal execution for a non-interactive functionality $\cF(x_1,\dots,x_n) \to (y_1,\dots,y_n)$. The environment $\cZ(1^\lambda,z)$ initializes dummy parties $\widetilde{P}_1,\dots,\widetilde{P}_n$ with $x_1,\dots,x_n$, and initializes a PPT $\Sim$ corrupting parties $[n] \setminus \cH$. 

\begin{enumerate}
    \item The honest dummy parties send their inputs $\{x_i\}_{i \in \cH}$ to the trusted third party. The simulator sends a set of inputs $\{x'_i\}_{i \in [n] \setminus \cH}$ on behalf of the corrupted dummy parties to the trusted third party.
    \item The third party computes $\cF$ on this set of inputs to obtain $(y_1,\dots,y_n)$, and returns the outputs to the dummy parties.
    \item The honest dummy parties send their outputs to $\cZ$ and $\Sim$ sends an output to $\cZ$.  
\end{enumerate}

Finally, $\cZ$ computes its own output, which is denoted by the random variable $\Ideal_{\cF,\Sim,\cZ}(1^\lambda,z)$. 

\paragraph{Security with selective abort} Our default notion of security will allow for \emph{selective aborts}: After the corrupted parties receive their output from the trusted third party, but \emph{before} the honest parties have received their output, the adversary / simulator may decide to send a list $\mathcal{L}$ of honest parties to the trusted third party. In this case, the third party will replace each of the honest party outputs $y_i$ for $i \in \mathcal{L}$ with the abort symbol $\bot$.

\paragraph{Ideal execution for interactive functionalities} We will also consider generalizing the notion of a functionality $\cF$ to allow for some limited interaction. In particular, we consider functionalities $\cF = (\cF_1,\cF_2)$ where $\cF_1(x_1,\dots,x_n) \to (y_{1,1},\dots,y_{1,n})$, and $\cF_2$ takes an additional public input $w$ and computes $\cF_2(x_1,\dots,x_n,w) \to (y_{2,1},\dots,y_{2,n})$. That is, the ideal execution for $\cF$ operates as follows. The environment $\cZ(1^\lambda,z)$ initializes dummy parties $\widetilde{P}_1,\dots,\widetilde{P}_n$ with $x_1,\dots,x_n$, and initializes a PPT $\Sim$ corrupting parties $[n] \setminus \cH$. 

\begin{enumerate}
    \item The honest dummy parties send their inputs $\{x_i\}_{i \in \cH}$ to the trusted third party. The simulator sends a set of inputs $\{x'_i\}_{i \in [n] \setminus \cH}$ on behalf of the corrupted dummy parties to the trusted third party.
    \item The third party computes $\cF_1$ on this set of inputs to obtain $(y_{1,1},\dots,y_{1,n})$, and returns the outputs to the dummy parties. 
    \item All outputs are forwarded to $\cZ$, who chooses $w$ and sends it to all dummy parties, who then forward $w$ to the trusted third party.
    \item The third party computes $\cF_2$ on this set of inputs to obtain $(y_{2,1},\dots,y_{2,n})$, and returns the outputs to the dummy parties. 
    \item The honest dummy parties send their outputs to $\cZ$ and the $\Sim$ sends an output to $\cZ$. 
\end{enumerate}

Finally, $\cZ$ computes its own output, which is denoted by the random variable $\Ideal_{\cF,\Sim,\cZ}(1^\lambda,z)$. 
Having defined both the real and ideal executions, we can now formally define a secure protocol.

\begin{definition}[Secure protocols]\label{def:secProto}

A protocol $\Pi$ securely computes a (potentially interactive) functionality $\cF$ in the presence of static malicious adversaries if for any PPT adversary $\Adv$, there exists a PPT simulator $\Sim$ such that for any PPT environment $\cZ$ with auxiliary input $z$,
\begin{equation*}
\left| \mathbb{P}(\Real_{\Pi,\Adv,\cZ}(1^{\lambda},z)=1) -  \mathbb{P}(\Ideal_{\cF,\Sim,\cZ}(1^\lambda,z) =1) \right| < \negl(\lambda),
\end{equation*}
where $\negl(\lambda)$ is asymptotically smaller than any polynomial in $\lambda$.
\end{definition}

That is, \Cref{def:secProto} means that a protocol is secure if the environment cannot distinguish between the real and ideal world executions. This implies that the adversary does not learn anything besides the output.  

\begin{remark}
This definition also captures correctness. For, if we consider an adversary/simulator that corrupts no parties, then the ideal execution is defined to honestly compute the functionality $\cF$ and output the result. Therefore, the real execution must also yield the correct result.
\end{remark}

Next, we define our first very simple ideal functionality, which we refer to as \emph{Simultaneous (All-to-All) Broadcast} $\cF_{\SB}$:
\begin{itemize}
    \item Take input $x_1,\dots,x_n$.
    \item Output $y_1 = (x_1,\dots,x_n),\dots,y_n = (x_1,\dots,x_n)$.
\end{itemize}

\begin{remark}\label{remark:all-to-all}
    If $\ell_{\mathsf{in}} = \max_{i \in [n]}\{|x_i|\}$, then from \Cref{subsection:overview-goldwasser-lindell} and the equality test of \Cref{lemma:equality} we can implement $\cF_{\SB}$ on point-to-point networks using $\tilde{O}(n^2\cdot (\ell_{\mathsf{in}}+ \lambda))$ total bits of communication, up to error $\negl(\lambda)$. 
\end{remark}

\paragraph{Known communication complexity results} Given an interactive functionality $\cF = (\cF_1,\cF_2)$, let $\ell_{\mathsf{in}} = \max_{i \in [n]}\{|x_i|\}$, let $D$ be the maximum circuit depth between $\cF_1$ and $\cF_2$, and let $\ell_{\mathsf{out}}$ be the \emph{total} number of bits of output, that is $\ell_{\mathsf{out}} = |y_{1,1}| + \dots |y_{1,n}| + |y_{2,1}| + \dots + |y_{2,n}|$. Known results \cite{MukherjeeW16,10.1007/978-3-030-26948-7_4} imply the following theorem; for convenience we will also provide a proof sketch.

\begin{theorem}\label{thm:MPC-from-FHE}
    Assuming the hardness of Learning with Errors (LWE), there exists a protocol for securely computing $\cF$ using one invocation of Simultaneous Broadcast $\cF_\SB$ on inputs of size $\poly(\lambda,D,\ell_{\mathsf{in}})$, plus an additional $\ell_{\mathsf{out}} \cdot n \cdot \poly(\lambda,D)$ bits of communication.
\end{theorem}

\begin{proof}[Proof sketch]
We will use the semi-malicious protocol of \cite{MukherjeeW16} based on (multi-key) fully-homomorphic encryption, combined with a universally composable non-interactive zero-knowledge  (NIZK) protocol from LWE (\cite{10.1007/978-3-030-26948-7_4,10.1007/978-3-031-22966-4_16}), in order to obtain a fully-malicious protocol.

In some detail, the first round consists of a simultaneous broadcast, where each party broadcasts (i) their public key, (ii) one ciphertext for each of their input bits, and (iii) a NIZK proof that they sampled all of this information honestly. Each of these components grows with the size of the lattice dimension, which is polynomial in the security parameter $\lambda$ and the depth $D$ of (largest) circuit to be computed. Thus, each party broadcasts a message of size at most $\poly(\lambda,D, \ell_{\mathsf{in}})$.

Next, for each bit of output across each of the two functionalities $\cF_1$ and $\cF_2$, we describe the communication necessary to deliver that bit of output to the party $i$ that is supposed to receive it. Each party $j \neq i$ must send to $i$ (in a point-to-point message), (i) a partial decryption corresponding to the output bit, and (ii) a NIZK proof that they computed their partial decryption share honestly. Each partial decryption share is a field element that grows polynomially in the size of the lattice dimension $\poly(\lambda,D)$. Moreover since the entire multi-key fully-homomorphic evaluation can be done publicly by each party, the only part of this computation that must be proven honest by the NIZK is the final (noisy) inner product between the $j$'th ciphertext component and party $j$'s secret key, which is also just the size of the lattice dimension $\poly(\lambda,D)$. Thus, the NIZK proof itself is of size $\poly(\lambda,D)$. So, we conclude that the communication necessary for each output bit is at most $n \cdot \poly(\lambda,D)$.

By combining the simultaneous broadcast step with the partial decryptions, we see that the entire protocol requires one simultaneous broadcast on inputs of size $\poly(\lambda,D,\ell_{\mathsf{in}})$, plus an additional $\ell_{\mathsf{out}} \cdot n \cdot \poly(\lambda,D)$ bits of communication.
\end{proof}

\begin{remark}
    By using \cite{10.1145/3566048} rather than \cite{MukherjeeW16}, we could improve the assumption from LWE to the more general assumption of (maliciously-secure) two-round oblivious transfer, at the cost of replacing the depth parameter $D$ with the \emph{size} $C$ of the largest circuit. That is, assuming maliciously-secure oblivious transfer, there exists a protocol that requires one invocation of simultaneous broadcast on inputs of size $\poly(\lambda, C, \ell_{\mathsf{in}})$, plus an additional $\ell_{\mathsf{out}} \cdot n \cdot \poly(\lambda, C)$ bits of communication.
\end{remark}

\paragraph{Encrypted functionality} We define an interactive ideal functionality that is parameterized by a public-key encryption scheme $\PKE = (\Gen,\Enc,\Dec)$ and a function $f\colon (x_1,\dots,x_m) \to \{0,1\}$ with $m$ inputs. This is, in fact, an input-less but \emph{randomized} functionality. Such randomized functionalities can be implemented by taking a random string $r_i$ as input from each party, and setting the random coins equal to $\bigoplus_{i \in [n]}r_i$. We make this explicit below.

\noindent \underline{$\cF[\PKE,f]$}

\begin{itemize}
    \item $\cF_{\Gen}[\PKE,f](r_1,\dots,r_n)$:
    \begin{itemize}
        \item Take input $r_1,\dots,r_n \in \{0,1\}^\lambda$, define $r \coloneqq \bigoplus_{i \in [n]} r_i$, and compute
        \begin{equation*}
            (\pk,\sk) \coloneqq \Gen(1^\lambda ; r).
        \end{equation*}
        \item Output $\pk$ to every party.
    \end{itemize}
    
    \item $\cF_{\Comp}[\PKE,f](r_1,\dots,r_n,w = (\ct_1,\dots,\ct_m))$:
    \begin{itemize}
        \item Take input $r_1,\dots,r_n \in \{0,1\}^\lambda$, define $r \coloneqq \bigoplus_{i \in [n]} r_i$, and compute
        \begin{equation*}
            (\pk,\sk) \coloneqq \Gen(1^\lambda ; r).
        \end{equation*}
        \item For each $i \in [m]$, compute $x_i \coloneqq \Dec(\sk,\ct_i)$.
        \item Compute $y = f(x_1,\dots,x_m)$ and output $y$ to every party.
    \end{itemize}
\end{itemize}

According to \Cref{thm:MPC-from-FHE} and the simultaneous broadcast protocol of \Cref{remark:all-to-all}, we can securely implement the above functionality with communication complexity $n^2 \cdot \poly(\lambda,D) + m \cdot n \cdot \poly(\lambda,D)$, where $D$ is the circuit depth of $f$.

\section{Communication-efficient protocols}\label{sec:upper}

In this section we prove our main result on the communication complexity of MPC with abort.

\begin{theorem}[Restatement of \Cref{theorem:results_new_protocol}]\label{theorem:formal_new_protocol} Assuming the hardness of LWE, there exists a protocol for MPC with abort against static malicious adversaries computing functions $f$ of depth $D$ using $O(n^2\cdot h^{-1}\poly(\lambda, D, \log n))$ bits of communication. The protocol has error $\negl(\lambda)$.
\end{theorem}

We refer the reader to \Cref{subsection:overview-optimal_communication} for an overview. In this section, we begin with a committee election protocol in \Cref{subsection:committee-election}, and then proceed with our MPC protocol in \Cref{subsection:mpc_with_aborts}.

\subsection{Committee election protocol}
\label{subsection:committee-election}

The overarching principle behind the committee election protocol is that by its conclusion, the parties agree on a small subset of the network which contains at least an honest party. We refer the reader to \Cref{alg:commitee_elect} for a description. 

\begin{algorithm}[ht]
    \setstretch{1.35}
    \caption{$\Com$, Committee Election Protocol.}
    \label{alg:commitee_elect}
    \KwInput{An integer $h\in [n]$, and integer security parameters $\alpha, \lambda$.}
    \KwOutput{Each party $i\in [n]$ receives a subset $C_i\subset [n]$, or aborts $\bot$.}

    \begin{algorithmic}[1]

    \State Each party $i\in [n]$ samples a bernoulli random variable with probability $p = \min(1, \alpha \cdot \frac{\log n}{h})$, resulting in a bit $b_i$.
    
    \State If $b_i=1$, then $i$ sends $b_i$ to all other $j\in [n]$.

    \State Let $C_i\subset [n]$ denote the bits $b_j=1$ received by $i\in [n]$ for $j\neq i$. If $|C_i|\geq 2\cdot p\cdot n$, abort.

    \State Each $i, j\in [n]$ s.t. $j\in C_i, i\in C_j$ engage in $\Eq_\lambda(C_i, C_j)$.
    
    \State If no pair rejects, then party $i$ ``receives'' $C_i$. Otherwise, abort.

    \end{algorithmic}
\label{proto:ComitElect}
\end{algorithm}

Step 2 (``election notification'') in \Cref{alg:commitee_elect} ensures that only the committee members notify the rest of the network of their election. Naturally, some malicious parties may attempt to lie about their election results or equivocate during the execution. However, step 3 ensures the total number of liars is bounded, and step 4 ensures that the committee members at least have a consistent view of each other, even in the presence of equivocation.

\begin{claim}\label{claim:committee-communication}
    The communication complexity of \Cref{alg:commitee_elect} is $\Tilde{O}(n^2\cdot h^{-1} \poly(\alpha, \lambda))$.
\end{claim}

\begin{proof}
    In an honest execution of the protocol, one readily inspects the total amount of bits sent within the committee is \linebreak $O(|C|^2\lambda \log n)$ and $O(|C|\cdot n)$ from the committee to the network. The upper bound on the committee size $|C|\leq n h^{-1} \alpha$ implies the claim.
\end{proof}

Correctness is based on the following simple ``Hitting Set'' Lemma, which follows immediately from a Chernoff bound.

\begin{lemma}\label{lemma:hitting_independent}
    Fix $H\subset [n]$. Let $R\subset [n]$ be a random subset, defined by independently sampling each element of $[n]$ with probability $p$. Then, $|H\cap R| \geq p\cdot |H|/2$ with probability $\geq 1-2^{-\Omega(p\cdot |H|)}$.
\end{lemma}

\begin{claim}
    After \Cref{alg:commitee_elect}, with probability $\geq 1-n^{-\Omega(\min(\alpha, \lambda))}$, either a party has aborted or 
    \begin{enumerate}
        \item At least 1 honest party $i$ was sampled in step 1, \textit{and}
        \item All the committee members agree on their view with $i$, $C_i\equiv C$.
    \end{enumerate}

   In the absence of any malicious parties, the protocol aborts with probability $\leq n^{2-\alpha}$.
    \end{claim}

        \begin{proof}
        Item 1 follows verbatim from \Cref{lemma:hitting_independent}. Item 2 follows from the security of the equality test \Cref{lemma:equality} and a union bound over all $\leq |C|^2\leq n^2$ tests. 
    \end{proof}

\subsection{MPC with abort}
\label{subsection:mpc_with_aborts}

We can now describe our protocol for MPC, \cref{alg:mpc_w_aborts}. As previously discussed, it is based on delegating the computation to a committee via the MPC functionality of \Cref{thm:MPC-from-FHE}. Succinct equality checks are run within the committee to ensure they all receive consistent views of the other parties inputs in the network.

\begin{algorithm}[ht]
    \setstretch{1.35}
    \caption{Multi-party Computation with abort.}
    \label{alg:mpc_w_aborts}
    \KwInput{Integer $h\in [n]$, a function $f \colon (\{0, 1\}^\ell)^n\rightarrow \{0, 1\}^{\ell'}$ and an input $x_i\in \{0, 1\}^\ell:$ $\forall i\in [n]$.}
    \KwOutput{Each party $i\in [n]$ outputs $f(x_1, \cdots x_n)$, or a party aborts $\bot$.}

    \begin{algorithmic}[1]

    \State Execute $\Com$, resulting in a committee $C\subset [n]$, and local views $C_i\subset [n]$ for $i\in [n]$.
    
    \State The committee generates $(pk, sk_c)_{c\in C}$ pairs using the encrypted functionality $\cF_\Gen$.

    \State Each party $c\in C$ forwards the public key $pk$ to all other $i\in [n]$. If any two messages differ, output $\bot$.

    \State Parties $i\in [n]$ encode their input $\ct_i = \Enc_{pk}(x_i)$, and send them to all $c\in C_i$.

    \State The committee members $c\in C$ concatenate their received messages $m_c = \{\ct_i\}_{i\in [n]}$, and pairwise run $\Eq_\lambda(m_{c'}, m_c)$. If the inputs are not all consistent, they abort $\bot$.

    \State The committee members engage in the encrypted functionality $\cF_\Comp$ with public inputs $\{\ct_i\}_{i\in [n]}$ and private inputs $\{sk_c\}_{c\in C}$, to compute the output $\out = f(x_1, \cdots, x_n)$.

    \State Finally, each committee member forwards $\out$ to all the members of the network. 

    \end{algorithmic}

\end{algorithm}

\Cref{theorem:formal_new_protocol} follows from the following \Cref{claim:communication-mpc,claim:correctness-mpc,claim:security-mpc}.
\begin{claim}\label{claim:communication-mpc}
The communication complexity of \Cref{alg:mpc_w_aborts} is 
    \begin{equation*}
        O(n^2h^{-1}\poly(\lambda, D, \log n)).
    \end{equation*}
\end{claim}
\begin{proof}
    Suppose the bitlength of the public keys, secret key shares, and ciphertexts $\ct$ is bounded by $b$. Then, the total amount of communication between network and committee, and within the committee, during steps 1 and 3, 4, 5 is $O(b\cdot (n\cdot |C| + \lambda|C|^2\log (n\cdot b))) \leq O(b n^2 h^{-1}\poly(\log n, \log b, \lambda))$ from the guarantees of \linebreak $\Com$ (\Cref{claim:committee-communication}) and $\Eq_\lambda$ (\Cref{lemma:equality}).

    In turn, from \Cref{thm:MPC-from-FHE}, the communication complexity of the encrypted functionality in steps 3 and 5 is $O(|C|\cdot n \cdot \poly(\lambda, D))$, and $b\leq \poly(\lambda, D)$. Under the bound $|C|\leq \Tilde{O}(\alpha nh^{-1})$ for the committee size, and the choice $\alpha=\lambda$ we conclude the desired claim.
\end{proof}

\begin{claim}\label{claim:correctness-mpc}
At the end of \Cref{alg:mpc_w_aborts}, with probability all but $n^{-\Omega(\lambda)}$, the $n$ parties have either all received the correct output or one of them has aborted. Moreover, if they are all honest, they abort with probability at most $n^{-\Omega(\lambda)}$.
\end{claim}

\begin{proof}
    The second part of the claim above---non-triviality---follows from the fact that $\Com$ and $\Eq$ individually fail with probability $n^{-\Omega(\lambda)}$, with honest or dishonest inputs. A union bound over the desired equality checks then gives us the claim. 

    In turn, conditioning on the success of $\Com$ and $\Eq$ as above, the resulting protocol is correct if the encrypted functionality $(\mathcal{F}_\Enc, \mathcal{F}_\Comp)$ is correct. Indeed, by \Cref{def:secProto} we prove the claim. 
\end{proof}

\begin{claim}\label{claim:security-mpc}
Assuming the hardness of learning with errors, \Cref{alg:mpc_w_aborts} is secure. 
\end{claim}

\begin{proof}
The proof of security is similarly inherited from the security of the encrypted functionality $\mathcal{F}$ of \Cref{def:secProto}. Indeed, an adversary which breaks the security of \Cref{alg:mpc_w_aborts} can be turned into an adversary for $\mathcal{F}$ by conditioning on the success of $\Com$ and $\Eq$---which as discussed occurs with probability $1-\negl(\lambda)$.
\end{proof}

\subsection{Generalization to multi-output functionalities}
Our protocol can be generalized to the setting where $n$ parties desire to compute a function that maps $n$ inputs to $n$ constant sized outputs (one output for each party).  
The main modifications in the MPC with abort protocol are the following:
\begin{enumerate}
    \item Each party encrypts their input as well as a randomly sampled secret key which will be used to encrypt their output. This is because each party should only learn its output. 
    \item The output length is now of order $n$ bits, so if every party in the committee forwards it to every other party in the network, the communication complexity in this setting would be $\mathcal{O}(\frac{n^3}{h^2})$. In order to avoid this blow-up in communication, each party's output is \emph{signed} using a digital signature after being encrypted. This is why we need 2 pairs of public and secret keys in our modified encrypted functionality (see below)---one pair is used to encrypt and decrypt the parties' inputs, the other pair is used to sign and verify signatures on the parties' outputs. From the security of digital signature schemes, the probability of forging a signature on a modified output is negligible. Assuming this security, it is enough to have any one party (even adversarial) forward the output to all the members in the network. 
\end{enumerate}

Before we formally state the protocol, we will define a modified  encrypt then authenticate functionality that is parameterized by a public-key encryption scheme $\PKE = (\Genenc,\Enc,\Dec)$, a digital signature scheme $\DS= (\Gensig,\Sign,\Vrfy)$, a secret key encryption scheme $\SKE = (\Genenc',\Enc',\Dec')$ and a function  $f \colon (x_1,\dots,x_m) 
\to (\{0,1\}^{l'})^{m}$ with $m$ inputs and $m$ constant sized outputs. This functionality is implemented by taking randomness and inputs from each party. 

\noindent \underline{$\cF_\Enc[\PKE,\SKE,\DS,f]$}
\begin{itemize}
    \item $\cF_{\Gen,1}[\PKE,\SKE,\DS,f](r_1,\dots,r_n)$:
    \begin{itemize}
        \item Take input $r_1,\dots,r_n \in \{0,1\}^\lambda$, define $r \coloneqq \bigoplus_{i \in [n]} r_i$, and compute
        \begin{equation*}
            (\pk,\sk) \coloneqq \Genenc(1^\lambda ; r).
        \end{equation*}
        \item Output $\pk$ to every party.
    \end{itemize}
    \item $\cF_{\Gen,2}[\PKE,\SKE,\DS,f]( r_1', \dots, r_n')$:
    \begin{itemize}
        \item Take input $r_1',\dots,r_n' \in \{0,1\}^\lambda$, define $r' \coloneqq \bigoplus_{i \in [n]} r_i'$, and compute
        \begin{equation*}
            (\pk',\sk') \coloneqq \Gensig(1^\lambda ; r').
        \end{equation*}
        \item Output $\pk'$ to every party.
    \end{itemize}
    \item $\cF_{\Comp,\Sign}[\PKE,\SKE,\DS,f](r_1,\dots,r_n, r_1', \dots, r_n', \linebreak w = (\ct_1,\dots,\ct_m),k'=(k_1', \dots k_n'))$:
    \begin{itemize}
        \item Take input $r_1,\dots,r_n \in \{0,1\}^\lambda$, define $r \coloneqq \bigoplus_{i \in [n]} r_i$, and compute
        \begin{equation*}
            (\pk,\sk) \coloneqq \Genenc(1^\lambda ; r).
        \end{equation*}
        \item Take input $r_1',\dots,r_n' \in \{0,1\}^\lambda$, define $r' \coloneqq \bigoplus_{i \in [n]} r_i'$, and compute
        \begin{equation*}
            (\pk',\sk') \coloneqq \Gensig(1^\lambda ; r').
        \end{equation*}
        \item For each $i \in [m]$, compute $x_i \coloneqq \Dec(\sk,\ct_i)$.
        \item For each $i \in [m]$, compute $k_i \coloneqq \Dec(\sk,k'_i)$.
        \item Compute $(y_1, \dots y_m) = f(x_1,\dots,x_m)$.
        \item For each $i \in [m]$, compute $ct_i' \coloneqq \Enc'(k_i,y_i)$.
        \item For each $i \in [m]$, compute $\sigma_i \gets \Sign(\sk',\ct_i') $ and output $((\ct_1',\sigma_1), \dots, (\ct_m',\sigma_m))$ to a single designated party.
    \end{itemize}
\end{itemize}
 
Finally, we formally describe the modified MPC protocol in \Cref{alg:mo_mpc_w_aborts}. 

\begin{algorithm}
    \setstretch{1.35}
    \caption{Multi-output Multi-party Computation with abort.}
    \label{alg:mo_mpc_w_aborts}
    \KwInput{Integer $h\in [n]$, a function $f \colon (\{0, 1\}^\ell)^n\rightarrow (\{0, 1\}^{\ell'})^n$, and an input $x_i\in \{0, 1\}^\ell$  for all $i\in [n]$.}
    \KwOutput{Each party $i\in [n]$ outputs $f(x_1, \cdots x_n)_i$, or aborts and outputs $\bot$.}

    \begin{algorithmic}[1]

    \State Execute $\Com$, resulting in a committee $C\subset [n]$, and local views $C_i\subset [n]$ for $i\in [n]$.
    
    \State The committee generates the encryption public key $\pk$ using the functionality $\cF_{\Gen,1}$.

    \State Each party $c\in C$ forwards the public key $pk$ to all other $i\in [n]$. If any two messages differ, output $\bot$.
    \State The committee generates the signature public key $\pk'$ using the functionality $\cF_{\Gen,2}$.
     \State Each party $c\in C$ forwards the public key $pk'$ to all other $i\in [n]$. If any two messages differ, output $\bot$.
    \State Parties $i \in [n]$ sample $r_i^*$, and compute $k_i \gets \Genenc'(1^{\lambda})$.
    \State Parties $i\in [n]$ encode their input $\ct_i \gets \Enc_{pk}(x_i)$, and their key $k_i'\gets\Enc_{pk}(k_i)$ and send them to all $c\in C_i$.

    \State The committee members $c\in C$ concatenate their received messages $m_c = \{\ct_i,k_i'\}_{i\in [n]}$, and pairwise run $\Eq_\lambda(m_{c'}, m_c)$. If the inputs are not all consistent, they abort $\bot$.

    \State The committee members engage in the encrypted functionality $\cF_{\Comp,\Sign}$ with public inputs $\{\ct_i, k'_i\}_{i\in [n]}$ to compute the output $\out = ((\ct_1',\sigma_1), \dots, (\ct_m',\sigma_m))$.

    \State Finally, one committee member forwards $(\ct_i',\sigma_i)$ to party $i$, for each $i \in [n]$.

    \State Each party $i\in [n]$ computes $\Vrfy(pk',\ct_i',\sigma_i)$. Party $i$ aborts if $\Vrfy(pk',\ct_i',\sigma_i)\neq 1$. Otherwise, each party $i\in [n]$ computes $y_i= \Dec'(k_i, \ct_i')$. 

    \end{algorithmic}

\end{algorithm}

\begin{acks}

We thank Shafi Goldwasser, Raluca Ada Popa, and Sukrit Kalra for discussions and several helpful comments. TB acknowledges support by the National Science Foundation Graduate Research Fellowship under Grant No. DGE 2146752. SK and OP acknowledge support by DARPA-ORACLEs under Grant No. 71654.

\end{acks}

\bibliographystyle{ACM-Reference-Format}
\bibliography{references_dblp}

\appendix

\section{Connectivity-based lower bound}\label{sec:lower}

In this section we prove our lower bound result. At the heart of our lower bound result is a \emph{connectivity} based argument, where we show that if the communication complexity is too small, then one honest party is doomed to be disconnected from all other honest parties, making it easy for the adversary to issue an attack. This idea is very similar to, and inspired by, a recent lower bound result for Broadcast by Blum et al.~\cite{BlumBCL23}.

\thmLB*

\begin{proof}
    We show that the lower bound holds even for Broadcast with abort, a special case of MPC. Let $\mathcal{A}$ be an arbitrary protocol solving Broadcast with abort, where the broadcasting party is $P$ with input $x$. For any party $Q$, let $N(Q)\subset [n]$ be the random subset of parties that communicate with $Q$ during an execution of $\mathcal{A}$. We claim that $\forall Q\in [n], \mathbb{E}[|N(Q)|]\geq \frac{n}{8(h-1)}$, which concludes our proof.

    Assume towards a contradiction that $\exists Q:\mathbb{E}[|N(Q)|]< \frac{n}{8(h-1)}$. Then the adversary can execute the following simple attack: the adversary decides that $Q$ is honest, picks $h-1$ other honest parties uniformly at random, and controls the rest of the parties. Let the set of honest parties be $H\subset [n]$. Since $\mathbb{E}[|N(Q)|]< \frac{n}{8(h-1)}$, by Markov's inequality, with probability at least $1/2$ we have $|N(Q)|<\frac{n}{4(h-1)}$. Conditioned on this event, since the adversary picks the other $h-1$ honest parties uniformly at random, the expected number of honest parties in $N(Q)$, $\mathbb{E}[|N(Q)\cap H||N(Q)< \frac{n}{8(h-1)}] \leq \frac{1}{4}$. Thereby, by another Markov's inequality, with probability at least $1/4$ none of the parties in $N(Q)$ are honest.

    Thus, with non-negligible probability, the honest node $Q$ is ``isolated'' and never communicates with any honest party. The adversary can now conclude the attack as follows:

    \begin{enumerate}
        \item If $Q = P$. Since all the other honest parties never communicate with the sender $P$, it can be impersonated by the malicious adversaries, making the other honest parties output a value different from $x$, violating correctness. 
        \item If $Q\neq P$. Then, $N(Q)$ can impersonate the rest of the network to $Q$, forcing them to output a value different from that of the sender $x'\neq x$, violating correctness.
    \end{enumerate}
\end{proof}

\section{Protocols with locality}
\label{section:local}

We dedicate this section to designing MPC protocols with locality. We present two protocols, which tradeoff communication complexity for locality. Both of which are based on replacing the communication on the complete graph for that on a sparse (random) routing network, and then engaging in either broadcast-with-aborts (as in \Cref{theorem:results-optimal-locality}), or our committee-based MPC protocol (in \Cref{theorem:results-local-tradeoff}) on that network. We begin by stating a theorem which achieves near-optimal locality, albeit with a far-from-optimal communication complexity:

\begin{theorem}[Near-optimal locality, restatement of \Cref{theorem:results-optimal-locality}]\label{theorem:optimal-locality}
    Assuming the hardness of LWE, for any $n\in \mathbb{N}, h \in [n]$ there exists a protocol for MPC with abort against static malicious adversaries, computing depth $D$ functions with constant-sized inputs using $O(n^3 h^{-1}\poly(\lambda, D, \log n))$ bits of communication and locality $O(\lambda nh^{-1} \log n)$. The protocol has error $\negl(\lambda)$.
\end{theorem}

We complement this result by combining its key ideas with our new communication-efficient protocols. In \Cref{theorem:local-tradeoff}, we design a protocol which sacrifices some locality---yet still is polynomially-sparser than the clique---while improving on the communication complexity of \Cref{theorem:optimal-locality}.

\begin{theorem}
 [A communication--locality tradeoff, restatement of \Cref{theorem:results-local-tradeoff}]
    \label{theorem:local-tradeoff}
    Assuming the hardness of LWE, for any $n\in \mathbb{N}, h \in [n]$ there exists a protocol for MPC with abort against static malicious adversaries, computing depth $D$ functions with constant-sized inputs using $O(n^3 h^{-3/2} \poly(\lambda, D, \log n)$ bits of communication and locality $O(\lambda n h^{-1/2}\log n)$. The protocol has error $\negl(\lambda)$.
\end{theorem}

\subsection{Establishing and communicating on a sparse routing network}

To begin, the parties in the network generate a routing network to communicate over a sparse graph---thereby achieving \textit{locality}.

\begin{algorithm}[ht]
    \setstretch{1.35}
    \caption{$\sparse$, Establishing a Sparse Routing Network.}
    \label{alg:routing}
    \KwInput{An integer $\log n < h<n$.}
    \KwOutput{Each party $i\in [n]$ receives a subset $N_i\subset [n]$, or aborts $\bot$.}

    \begin{algorithmic}[1]

    \State Each party $i\in [n]$ samples a uniformly random subset $N^{out}_i\subset [n]$ of $\alpha\cdot n h^{-1} \cdot \log n$ parties.
    
    \State Each party $i\in [n]$ notifies $j\in N_i$. Let $N^{in}_i\subset [n]$ denote the notifications $i$ received. 

    \State If $\frac{1}{2}\cdot \alpha\cdot n h^{-1} \cdot \log n \leq |N_i^{in}|\leq 2\cdot \alpha\cdot n h^{-1} \cdot \log n$, $i$ outputs $\bot$. Otherwise, $N_i = N_i^{in}\cup N_i^{out}$.
    \end{algorithmic}
\end{algorithm}

\begin{claim}\label{claim:connected}
    Fix a subset $H\subset [n]$ of $h$ honest parties. After running $\sparse$, with probability all but $n^{-\Omega(\alpha)}$ either the parties abort or

    \begin{itemize}
        \item The maximum degree of the resulting graph $G$ is $O(\alpha n h^{-1}\log n)$, 
        \item The subgraph of the resulting graph $G$ induced by $H$ is connected.
    \end{itemize}
    Moreover, in the absence of malicious parties, the parties abort with probability $\leq n^{-\Omega(\alpha)}$.
\end{claim}

\begin{proof}
    The upper bound on the degree of the graph is implied by step 3 of \Cref{alg:routing}. Next, we turn to establish the non-triviality condition. In expectation, the number of incoming edges for a fixed node is $\mathbb{E}[|N_i^{in}|] \equiv d =  \alpha n h^{-1}\log n$. By a Chernoff bound, the probability $|N_i^{in}|$ exceeds twice its expectation is $n^{-\Omega(\alpha)}$. By a union bound, in the absence of dishonest parties, the probability the parties abort is $\leq n\cdot n^{-\Omega(\alpha)}$.

    Now, let us turn to the connectivity of honest parties in the network. It suffices to argue that each honest party is connected to at least $h/2 + 1$ other honest parties. The pigeonhole principle will then imply that any two parties are connected. Fix an honest party $H_0$, and consider the following probabilistic process:
    
    \begin{enumerate}
        \item Initialize $i=1$. Initialize the set of permitted parties $S$ to all honest parties but $H_0$, $\coloneqq \cH \setminus \{H_0 \}$.
        \item Choose a set $R \subseteq [n]$ of size $d$ uniformly at random. If $R \cap S = \emptyset$, halt. Else, let $H_i$ be any party in $R \cap S$.
        \item Update $i \coloneqq i + 1$, and $S \coloneqq S \setminus \{ H_i \}$.
    \end{enumerate}
    
    Note that this process is a different view on the very same process by which hops are formed in the protocol. We claim that the probability that this process halts before $h/2 + 1$ rounds is at most $\negl(n)$. Indeed, observe that for any $i \leq h/2$, the probability that the process halts at round $i$ (conditioned on not halting in prior rounds) is
    \begin{equation*} 
        \left( \frac{n-h + i}{n} \right) ^ {d} < (1 - \frac{h}{2n})^{d} \leq 2^{-\Theta(h \cdot d/n)} \leq n^{-\Omega(\alpha)}
    \end{equation*}
    Two more union bounds (one over the $O(h)$ rounds, and another over at most $h$ parties) conclude the claim.

\end{proof}

The $\gossip$ protocol below shows how to implement the simultaneous broadcast functionality over the graph $G$ generated using $\sparse$. In a nutshell, the parties repeatedly forward their received messages until each message has been forwarded by each node once. Although this protocol is local, it naturally comes at a communication cost.

\begin{algorithm}[ht]
    \setstretch{1.35}
    \caption{$\gossip$, A Sparse Broadcast Protocol.}
    \label{alg:gossip}
    \KwInput{Private inputs $\{x_i\}_{i\in [n]}$, and a graph $G$.}
    \KwOutput{Each party $j\in [n]$ receives every other $x_i:i\in [n]$, or aborts $\bot$.}

    \begin{algorithmic}[1]

    \State Assuming $x_i\neq \nulll$, each party $i\in [n]$ sends their input, $(x_i, i)$ together with their ID to their neighbors on $G$.

    \State Each party $i\in [n]$ repeatedly forwards their received messages to their neighbors, without forwarding twice messages from the same ID.

    \State If any two received messages with same ID $i$ differ, then output $\bot$.
    
    \end{algorithmic}
\end{algorithm}

\begin{claim}\label{claim:gossip}
    Let $k$ out of the $n$ parties be given inputs of length $\ell$. After running $\gossip$, with probability all but $n^{-\Omega(\alpha)}$

    \begin{itemize}
        \item Correctness: Either the parties reject, or the honest parties agree on their outputs.
        \item Non-triviality: In the absence of malicious parties, all parties agree on their outputs. 
        \item Communication complexity: The total amount of bits of communication sent is $O(\alpha k\cdot  h^{-1} n^2\cdot \poly(l, \log n))$.
    \end{itemize}
\end{claim}

\begin{proof}
    From \Cref{claim:connected}, we ensure that the graph $G$ connects all the honest parties. Then, unless a party detects an equivocation and rejects, the honest parties successfully receive the inputs of the other honest parties. The communication length follows from a bound on the total number of edges in the graph.
\end{proof}

Our protocol for MPC with abort with optimal locality simply applies this local protocol for simultaneous broadcast to our encrypted functionality of \Cref{thm:MPC-from-FHE}. This concludes the proof of \Cref{theorem:optimal-locality}.

\subsection{Local committee election}

Using the local communication protocols developed in the previous subsection, we can now establish a local committee election protocol (in \Cref{alg:localcommittee}).

\begin{algorithm}[ht]
    \setstretch{1.35}
    \caption{$\localcom$, Local Committee Election.}
    \label{alg:localcommittee}
    \KwInput{An integer $h\in [n]$.}
    \KwOutput{Each party $i\in [n]$ receives a subset $C_i\subset [n]$, or aborts $\bot$.}

    \begin{algorithmic}[1]

    \State Run $\sparse$, resulting in a graph $G$.
    
    \State Each party $i\in [n]$ flips a coin with bias $p = \min(1, \alpha \cdot \frac{\log n}{h^{1/2}})$, resulting in a bit $b_i\in \{0, 1\}$.
    
    \State Run $\gossip$, where parties with $b_i=0$ have $\nulll$ input. 

    \State Let $C_i\subset [n]$ denote the bits $b_j=1$ received by $i\in [n]$ for $j\neq i$. If $|C_i|\geq 2\cdot p\cdot n$, output $\bot$.

    \State Each $i, j\in [n]$ s.t. $j\in C_i$ and $i\in C_j$ engage in $\Eq_\lambda(C_i, C_j)$. 

    \State If no pair rejects, then party $i$ ``receives'' $C_i$. Otherwise, output $\bot$.
    \end{algorithmic}
\end{algorithm}

\begin{claim}\label{claim:local_committee}
    After \Cref{alg:localcommittee}, then with probability $1-n^{-\Omega(\alpha)}$, either a party has aborted or 

    \begin{enumerate}
        \item $\geq \frac{1}{2}\alpha h^{1/2} \log n$ honest parties were sampled in step 1, \textit{and}
        \item All the honest parties sampled in step 1 agree on their view $C_i\equiv C$.
    \end{enumerate}

    Moreover, $|C|\leq 2\alpha nh^{-1/2} \log n$, and the total amount of bits communicated is $\alpha^2 n^3 h^{-3/2}\log^3 n$.
    \end{claim}

    \begin{proof}
        Item 1 follows from the hitting set \Cref{lemma:hitting_independent}. Item 2 follows from the receival guarantees of \Cref{claim:gossip} and its bit complexity, together with that of the equality test in \Cref{lemma:equality}.
    \end{proof}

\subsection{MPC with abort}

To conclude this section, we build on our prior results on MPC to establish a protocol which tradeoffs communication complexity for locality. Its description is presented in \Cref{alg:mpc-local}, and a formal statement of its guarantees in \Cref{theorem:local-tradeoff}.

\begin{algorithm}[ht]
    \setstretch{1.35}
    \caption{A Local Protocol for Multi-party Computation with abort.}
    \label{alg:mpc-local}
    \KwInput{Integer $\log n < h<n$, a function $f:(\{0, 1\}^\ell)^n\rightarrow \{0, 1\}^{\ell'}$ and an input $x_i\in \{0, 1\}^\ell:$ $\forall i\in [n]$.}
    \KwOutput{Each party $i\in [n]$ outputs $f(x_1, \cdots x_n)$, or a party aborts $\bot$.}

    \begin{algorithmic}[1]

    \State Execute $\localcom$, resulting in a committee $C\subset [n]$, and local views $C_i\subset [n]$ for $i\in [n]$.
    
    \State The committee generates $(pk, sk_c)_{c\in C}$ pairs using the encrypted functionality $\cF_\Gen$.

    \State Each committee member $c\in C$ chooses a uniformly random subset $S_c\subset [n]$ of size $n h^{-1/2}$.

    \State Each $c\in C$ forwards the public key $pk$ to all $i\in S_c$. If any two messages differ, output $\bot$.

    \State Each $i\in [n]$ encodes their input $\ct_i = \Enc_{pk}(x_i)$, and send them to all $c$ s.t. $i\in S_c$.

    \State Each $c\in C$ concatenates their received messages $\ct_{S_c} = \{(i, \ct_i)\}_{i\in [S_c]}$ together with sender information, and forwards them to all other $c'\in C$. If any two received ciphers $(i, \ct_i)$ differ for a same $i\in [n]$, output $\bot$.

    \State Each $c\in C$ concatenates their received messages $m_c = \{\ct_i\}_{i\in [n]}$ and pairwise run $\Eq_\lambda(m_{c'}, m_c)$. If they don't all accept, output $\bot$.

    \State The committee members engage in the encrypted functionality $\cF_\Comp$ with public inputs $\{\ct_i\}_{i\in [n]}$ and private inputs $\{sk_c\}_{c\in C}$, to compute the output $\out = f(x_1, \cdots, x_n)$.

    \State Finally, each $c\in C$ forwards $\out$ to all the members of $S_c$. If any two received messages differ, they abort.
    \end{algorithmic}

\end{algorithm}

At a high level, \Cref{alg:mpc-local} attempts to leverage the communication-efficiency of the committee-based MPC protocols we developed in \Cref{sec:upper}, however, implemented in a local manner. To do so, we need to ensure that the committee members themselves communicate with few members of the rest of the network, and yet, each member of the network still communicates with at least an honest committee member. For this purpose, in step 3 of \Cref{alg:mpc-local}, the Committee members randomly (and obliviously) partition the network into (overlapping) subsets.

To later prove security and correctness of our scheme, we begin with a ``covering'' claim, which both stipulates that each the parties in the network are connected to honest parties in the committee and that their

\begin{claim}\label{claim:honest_connected}
    With probability all but $1-n^{-\Omega(\alpha)}$, by the end of step 3 of \Cref{alg:mpc-local}, either the parties have aborted or for every $i\in [n]\setminus C$ there exists an honest committee member $c\in C\cap H$ such that $i\in S_c$. 
\end{claim}

\begin{proof}
    From \Cref{claim:local_committee}, with probability all but $1-n^{-\Omega(\alpha)}$, there are $\geq \frac{1}{2}\alpha h^{1/2} \log n$ honest parties in the committee. Now, fix $i\in [n]$. The probability $i\in S_c$ is $|S_c|/n = h^{-1/2}$. Thereby, 

    \begin{align}
        \mathbb{P}[\nexists c\in C\cap H: i\in S_c]&\leq \bigg(1- \frac{|S|}{n}\bigg)^{\frac{1}{2}\alpha h^{1/2} \log n}\\
        &\leq
        e^{-\Omega(\alpha \cdot \log n)}
    \end{align}

    \noindent A union bound over $i\in [n]$ concludes the argument.
\end{proof}

To prove \Cref{theorem:local-tradeoff}, we proceed in a sequence of 3 claims, which analyze the locality, the communication complexity, the correctness, and finally the security of \Cref{alg:mpc-local}.

\begin{claim}\label{claim:local-locality}
    The \emph{locality} of \Cref{alg:mpc-local} is $O(\lambda \cdot  n h^{-1/2}\cdot \log n)$.
\end{claim}

\begin{proof}
    The locality of the protocol is the number of parties each node communicates with during the $\localcom$ protocol, within the committee, and during the network partition. In this manner, it is bounded by 
    \begin{equation*}
        (\text{Degree of }G) + |S_c| + |C|\leq O(\alpha\cdot n\cdot h^{-1/2}\log n)
    \end{equation*}

    Where we extract the properties of $G$ from \Cref{claim:connected}. The choice of $\alpha = \lambda$ concludes the claim. 
\end{proof}

\begin{claim}\label{claim:local-comm}
    The \emph{communication complexity} of \Cref{alg:mpc-local} is
    \begin{equation*}
        O(n^3 h^{-3/2}\cdot \poly(\lambda, \log n, D))
    \end{equation*}
\end{claim}

\begin{proof}
    From \Cref{claim:local_committee}, the communication complexity of the local committee election protocol is $\alpha^2 n^3 h^{-3/2}\log^3 n$. If the size of the keys $\pk, \sk_i$, the ciphertexts $\ct_i$, and the output $\out$ are bounded by $b$, then steps 4, 5, and 9 require $O(\max(b, \log n)\cdot |C|\cdot n)$ bits of communication, due to the network-committee communication. In turn, steps 6 requires $O(\max(b, \log n)\cdot |C|^2\cdot |S|_c)$, and, from the properties of the equality check in \Cref{lemma:equality}, step 7 requires $O(\lambda\cdot \log n \cdot |C|^2)$ bits. 

     In turn, from \Cref{thm:MPC-from-FHE}, the communication complexity of the encrypted functionality in steps 2 and 8 is $O(|C|\cdot n \cdot \poly(\lambda, D))$, and $b\leq \poly(\lambda, D)$. Under the bound $|C|\leq \Tilde{O}(\alpha nh^{-1/2})$ for the committee size, the bound $|S_c|\leq nh^{-1/2}$ on the network partitions, and the choice $\alpha=\lambda$ we conclude the desired claim.
    
\end{proof}

\begin{claim}\label{claim:local-correct-and-secure}
Assuming the hardness of LWE, \Cref{alg:mpc-local} is correct, and secure with error $\negl(\lambda)$.
\end{claim}

\begin{proof}
    The event that either the parties abort, or the local committee election protocol succeeds (\Cref{claim:local_committee}), that every party $i\in [n]$ in the network is connected to at least one honest party in the committee (\Cref{claim:honest_connected}), and all the equality tests are correct (\Cref{lemma:equality}), occurs with probability all but $n^{-\Omega(\alpha)}$ from \Cref{claim:honest_connected} and a union bound. 

    The proof of correctness and security of the protocol is now analogous to the proof of \Cref{claim:security-mpc}. Indeed, an adversary which breaks the security of \Cref{alg:mpc-local} can be turned into an adversary for $\mathcal{F}$ by conditioning on the events above. The choice of $\alpha = \lambda$ concludes the claim.
\end{proof}

Put together, \Cref{claim:local-locality,claim:local-comm,claim:local-correct-and-secure} conclude the proof of \Cref{theorem:local-tradeoff}.
\end{document}